\pdfoutput=1
\let\FULLVERSION\relax
\documentclass[a4paper,UKenglish,cleveref,thm-restate]{lipics-v2021}
\nolinenumbers
\hideLIPIcs
\usepackage{newunicodechar}
\usepackage{silence}
\newunicodechar{♢}{\tikz \node[inner sep=1.5,draw,diamond] {};}
\newunicodechar{☆}{\tikz \node[inner sep=1,draw,star,star point ratio=2] {};}
\newunicodechar{△}{\kern 0.5pt\tikz \node[inner sep=1.2,draw,regular polygon,regular polygon sides=3] {};\kern 0.5pt}
\newunicodechar{⬜}{\kern 0.5pt\tikz \node[inner sep=1.7,draw,regular polygon,regular polygon sides=4] {};\kern 0.5pt}
\newunicodechar{◯}{\tikz[baseline=-3pt] \node[inner sep=1.7,draw,cloud,cloud puffs=4,cloud puff arc=190] {};}

\newunicodechar{⊥}{\bot}
\newunicodechar{•}{\item}
\newunicodechar{✓}{\checkmark}
\newunicodechar{✗}{\xmark}
\newunicodechar{…}{\dots}
\newunicodechar{≔}{\coloneqq}
\newunicodechar{⁻}{^-}
\newunicodechar{⁺}{^+}
\newunicodechar{₋}{_-}
\newunicodechar{₊}{_+}
\newunicodechar{ℓ}{\ell}
\newunicodechar{•}{\item}
\newunicodechar{…}{\dots}
\newunicodechar{≔}{\coloneqq}
\newunicodechar{≤}{\leq}
\newunicodechar{≥}{\geq}
\newunicodechar{≰}{\nleq}
\newunicodechar{≱}{\ngeq}
\newunicodechar{⊕}{\oplus}
\newunicodechar{⊗}{\otimes}
\newunicodechar{≠}{\neq}
\newunicodechar{¬}{\neg}
\newunicodechar{≡}{\equiv}
\newunicodechar{₀}{_0}
\newunicodechar{₁}{_1}
\newunicodechar{₂}{_2}
\newunicodechar{₃}{_3}
\newunicodechar{₄}{_4}
\newunicodechar{₅}{_5}
\newunicodechar{₆}{_6}
\newunicodechar{₇}{_7}
\newunicodechar{₈}{_8}
\newunicodechar{₉}{_9}
\newunicodechar{⁰}{^0}
\newunicodechar{¹}{^1}
\newunicodechar{²}{^2}
\newunicodechar{³}{^3}
\newunicodechar{⁴}{^4}
\newunicodechar{⁵}{^5}
\newunicodechar{⁶}{^6}
\newunicodechar{⁷}{^7}
\newunicodechar{⁸}{^8}
\newunicodechar{⁹}{^9}
\newunicodechar{∈}{\in}
\newunicodechar{∉}{\notin}
\newunicodechar{⊂}{\subset}
\newunicodechar{⊃}{\supset}
\newunicodechar{⊆}{\subseteq}
\newunicodechar{⊇}{\supseteq}
\newunicodechar{⊄}{\nsubset}
\newunicodechar{⊅}{\nsupset}
\newunicodechar{⊈}{\nsubseteq}
\newunicodechar{⊉}{\nsupseteq}
\newunicodechar{∪}{\cup}
\newunicodechar{∩}{\cap}
\newunicodechar{∀}{\forall}
\newunicodechar{∃}{\exists}
\newunicodechar{∄}{\nexists}
\newunicodechar{∨}{\vee}
\newunicodechar{∧}{\wedge}
\newunicodechar{ℝ}{\mathbb{R}}
\newunicodechar{𝔽}{\mathbb{F}}
\newunicodechar{ℕ}{\mathbb{N}}
\newunicodechar{𝔼}{\mathbb{E}}
\newunicodechar{𝟙}{\mathds{1}}
\newunicodechar{𝒪}{\mathcal{O}}
\newunicodechar{ℤ}{\mathbb{Z}}
\newunicodechar{·}{\cdot}
\newunicodechar{∘}{\circ}
\newunicodechar{×}{\times}
\newunicodechar{↑}{\uparrow}
\newunicodechar{↓}{\downarrow}
\newunicodechar{→}{\rightarrow}
\newunicodechar{←}{\leftarrow}
\newunicodechar{↪}{\hookrightarrow}
\newunicodechar{⇒}{\Rightarrow}
\newunicodechar{⇐}{\Leftarrow}
\newunicodechar{↔}{\leftrightarrow}
\newunicodechar{⇔}{\Leftrightarrow}
\newunicodechar{↦}{\mapsto}
\newunicodechar{∅}{\varnothing}
\newunicodechar{∞}{\infty}
\newunicodechar{≅}{\cong}
\newunicodechar{≈}{\approx}
\newunicodechar{ℓ}{\ell}
\newunicodechar{⌊}{\lfloor}
\newunicodechar{⌋}{\rfloor}
\newunicodechar{⌈}{\lceil}
\newunicodechar{⌉}{\rceil}

\newunicodechar{α}{\alpha}
\newunicodechar{β}{\beta}
\newunicodechar{γ}{\gamma}
\newunicodechar{Γ}{\Gamma}
\newunicodechar{δ}{\delta}
\newunicodechar{Δ}{\Delta}
\newunicodechar{ε}{\varepsilon}
\newunicodechar{ζ}{\zeta}
\newunicodechar{η}{\eta}
\newunicodechar{θ}{\theta}
\newunicodechar{Θ}{\Theta}
\newunicodechar{ι}{\iota}
\newunicodechar{κ}{\kappa}
\newunicodechar{λ}{\lambda}
\newunicodechar{Λ}{\Lambda}
\newunicodechar{μ}{\mu}
\newunicodechar{ν}{\nu}
\newunicodechar{ξ}{\xi}
\newunicodechar{Ξ}{\Xi}
\newunicodechar{π}{\pi}
\newunicodechar{Π}{\Pi}
\newunicodechar{ρ}{\rho}
\newunicodechar{σ}{\sigma}
\newunicodechar{Σ}{\Sigma}
\newunicodechar{τ}{\tau}
\newunicodechar{υ}{\upsilon}
\newunicodechar{ϒ}{\Upsilon}
\newunicodechar{φ}{\phi}
\newunicodechar{ϕ}{\varphi}
\newunicodechar{Φ}{\Phi}
\newunicodechar{χ}{\chi}
\newunicodechar{ψ}{\psi}
\newunicodechar{Ψ}{\Psi}
\newunicodechar{ω}{\omega}
\newunicodechar{Ω}{\Omega}

\usepackage{xcolor}
\usepackage{xspace}
\usepackage{booktabs}
\usepackage{amsmath,amsthm,amssymb,mathtools}
\usepackage{dsfont}
\usepackage{bm}
\usepackage{enumerate}
\usepackage{pifont}\def\xmark{\ding{55}} 
\usepackage{thmtools}
\usepackage{thm-restate}
\usepackage[linesnumbered,commentsnumbered,vlined]{algorithm2e}
\DontPrintSemicolon
\LinesNumbered
\crefname{algocf}{Algorithm}{Algorithms}
\crefname{algorithm}{Algorithm}{Algorithms}
\crefname{algocf}{Algorithm}{Algorithms}
\crefname{algocfline}{Line}{Lines}

\SetCommentSty{commentFont}
\SetKwProg{algo}{Algorithm}{:}{}
\SetKw{Continue}{continue}
\SetKw{Break}{break}
\SetKw{With}{with}
\SetKw{Retry}{retry}
\usepackage{tikz}
\usetikzlibrary{shapes.geometric,decorations.pathmorphing,arrows.meta,calc,patterns}

\def\refrel#1#2#3{\stackrel{\text{#2 \ref{#1}}}{#3}}
\def\eqRefRel#1#2#3{\stackrel{\text{#2(\ref{#1})}}{#3}}
\def\itemRef#1{\lipicsLabel{(\ref{#1})}}
\def\reasonrel#1#2{\stackrel{\text{#1}}{#2}}
\def\myparagraph#1{\smallskip\noindent\textbf{#1}}
\let\paragraph=\myparagraph

\def\REPT{REP'\xspace}
\def\REP{REP\xspace}
\def\Exp{\mathrm{Exp}}

\def\Bin{\mathrm{Bin}}

\def\e{\mathrm{e}}
\def\cp{c^{△}}
\def\cpk{c^{△}_{k,1}}
\def\cpkl{c^{△}_{k,ℓ}}
\def\ck{c^{*}_{k,1}}
\def\c{c^*}
\def\ckl{c^{*}_{k,ℓ}}
\def\moves{\mathrm{moves}}
\def\peel{\mathrm{peel}}

\newcommand\lipicsLabel[1]{\textcolor{darkgray}{\sffamily\upshape\bfseries\mathversion{bold}#1}}
\usepackage{pdfpages}
\usepackage[export]{adjustbox}
\newif\ifanonymous\anonymousfalse
\newif\iffullversion
\ifdefined\FULLVERSION
  \fullversiontrue
\else
  \fullversionfalse
\fi

\hypersetup{colorlinks=false}

\title{Insertion Time of Random Walk Cuckoo Hashing\texorpdfstring{\\}{}below the Peeling Threshold\iffullversion\else\ (extended abstract)\fi}
\titlerunning{Insertion Time of Random Walk Cuckoo Hashing below the Peeling Threshold}
\ifanonymous
  \author{Undisclosed Author}{Undisclosed University}{}{}{}
  \authorrunning{Undisclosed Author}
\else
  \author{Stefan Walzer}{University of Cologne}{walzer@cs.uni-koeln.de}{}{DFG grant WA 5025/1-1.}
  \authorrunning{S.\ Walzer}
\fi


\begin{document}
\maketitle


\begin{abstract}
    Most hash tables have an insertion time of $𝒪(1)$, often qualified as “expected” and/or “amortised”. While insertions into cuckoo hash tables indeed seem to take $𝒪(1)$ expected time in practice, only polylogarithmic guarantees are proven in all but the simplest of practically relevant cases. Given the widespread use of cuckoo hashing to implement compact dictionaries and Bloom filter alternatives, closing this gap is an important open problem for theoreticians.
    
    In this paper, we show that random walk insertions into cuckoo hash tables take $𝒪(1)$ expected amortised time when any number $k ≥ 3$ of hash functions is used and the load factor is below the corresponding \emph{peeling threshold} (e.g.\ $≈0.81$ for $k = 3$). To our knowledge, this is the first meaningful guarantee for constant time insertion for cuckoo hashing that works for $k ∈ \{3,…,9\}$.
    
    In addition to being useful in its own right, we hope that our key-centred analysis method can be a stepping stone on the path to the true end goal: $𝒪(1)$ time insertions for all load factors below the \emph{load threshold} (e.g. $≈0.91$ for $k = 3$).
\end{abstract}

\newpage




\section{Introduction}

\paragraph{Cuckoo Hashing Basics.} Cuckoo hashing is an elegant approach for constructing compact and efficient dictionaries that has spawned both a rich landscape of theoretical results and popular practical applications. Briefly, each key $x$ in a set of $m$ keys is assigned $k$ positions $h₁(x),…,h_k(x)$ in an array of $n ≥ m$ buckets via hash functions $h₁,…,h_k$. Each bucket can hold at most $ℓ$ keys and the challenge is to choose for each key one of its assigned buckets while respecting bucket capacities. We follow previous literature in assuming $k ≥ 2$ and $ℓ ≥ 1$ to be constants and $h₁,…,h_k$ to be uniformly random functions.

Since the term cuckoo hashing was coined for $(k,ℓ) = (2,1)$ \cite{PR:Cuckoo:2004} and then generalised to $k ≥ 3$ \cite{FPSS:Space_Efficient:2005} and $ℓ ≥ 2$ \cite{DW07:Balanced:2007}, a major focus of theory papers has been to determine the load thresholds $\ckl$, which are constants such that for a load factor $\frac{m}{n} < \ckl-ε$ a placement of all keys exists with high probability (whp, defined in this paper as probability $1-n^{-Ω(1)}$) and for $\frac{m}{n} > \ckl+ε$ a placement does not exist whp. This project has since been completed \cite{FR:The_k-orientability:2007,CSW:The_Random:2007,DGMMPR:Tight:2010,FP:Sharp:2012,FKP:The_Multiple:2011,L:A_New_Approach:2012} and we reproduce some thresholds in \cref{tab:thresholds} for reference. Further research pursued cuckoo hashing variants with reduced failure probability \cite{KMW:More_Robust:2009}, improved load thresholds \cite{DW07:Balanced:2007,LP:3.5-Way:2009,Walzer:OverlappingBlocks:2018} or weaker randomness requirements for $h₁,…,h_k$ \cite{MT:Peeling_Arguments:2012,ADW:Explicit:2013,Leconte:Cuckoo:2013}. Moreover, cuckoo filters \cite{FAK:CuckooFilterBetter:2013,FAKM:CuckooFilterReallyBetter:2014,Eppstein:CuckooFilter:2016}, which are Bloom filter alternatives based on cuckoo hashing, are now widely popular in the data base community.

\begin{table}[tbh]
  \def\oneo{1\phantom{.000}}
  \begin{tabular}{cccccccc}
      \toprule 
      $ℓ$\textbackslash\raisebox{1.5pt}{$k$} & 2 & 3 & 4 & 5 & 6 & 7\\
      \midrule
      1 & \phantom{0.0}–\phantom{0}/0.500 & 0.818/0.918 & 0.772/0.977 & 0.702/0.992 & 0.637/0.997 & 0.582/0.999\\
      2 & 0.838/0.897 & 0.776/0.988 & 0.667/0.998 & 0.579/\oneo & 0.511/\oneo & 0.457/\oneo\\
      3 & 0.858/0.959 & 0.725/0.997 & 0.604/\oneo & 0.515/\oneo & 0.450/\oneo & 0.399/\oneo\\
      4 & 0.850/0.980 & 0.687/0.999 & 0.562/\oneo & 0.476/\oneo & 0.412/\oneo & 0.364/\oneo\\
      5 & 0.837/0.990 & 0.658/\oneo & 0.533/\oneo & 0.448/\oneo & 0.387/\oneo & 0.341/\oneo\\
      6 & 0.823/0.994 & 0.635/\oneo & 0.511/\oneo & 0.427/\oneo & 0.368/\oneo & 0.323/\oneo\\
      \bottomrule
  \end{tabular}
  \caption{For each $2 ≤ k ≤ 7$ and $1 ≤ ℓ ≤ 6$ a cell shows $\cpkl/ℓ$ (left) and $\ckl/ℓ$ (right), rounded to three decimal places. The thresholds are divided by $ℓ$ to reflect the corresponding memory efficiency (populated space over allocated space).}
  \label{tab:thresholds}
\end{table}

\paragraph{Cuckoo Insertions.} An important concern in all variants of cuckoo hashing is how to insert a new key $x$ into an existing data structure. If all buckets $h₁(x),…,h_k(x)$ are full, then one key that is currently placed in those buckets has to be moved out of the way into one of its alternative buckets, which might require additional relocations of keys. Two natural strategies for organising insertions have been proposed \cite{FPSS:Space_Efficient:2005}. Breadth-first search (BFS) insertion systematically pursues all possibilities for relocating keys in parallel. Random walk (RW) insertion starts by optimistically placing $x$ into bucket $h_i(x)$ for a uniformly random $i ∈ [k]$ (where in this paper $[a] := \{1,…,a\}$ for $a ∈ ℕ$), without first considering any of the $k-1$ other buckets. If $h_i(x)$ was already full, then a random key is evicted from $h_i(x)$ and is itself placed into one of its $k-1$ alternative buckets. This chain of evictions continues until a bucket with leftover space is reached (see \cref{algo:rwi} for the case with $ℓ = 1$).

\begin{figure}[htb]
  \def\iold{i_{\mathrm{old}}}
  \begin{algorithm}[H]
    \algo{\normalfont RW($x$)}{
      $\iold ← ⊥$\;
      \Repeat{$x = ⊥$}{
        pick random $i ∈ \{h₁(x),…,h_k(x)\} \setminus \{\iold\}$\;
        swap($x,B[i]$)\;
        $\iold ← i$\;
      }
    }
  \end{algorithm}
  \caption{The random walk insertion algorithm for $ℓ = 1$. The array $B$ of buckets is initialised with $⊥$.}
  \label{algo:rwi}
\end{figure}

Experiments suggest that, regardless of $k$ and $ℓ$, and for any load factor $\frac{m}{n} < \ckl - ε$ where insertions still succeed whp, the expected insertion time is independent of $n$, hence “$𝒪(1)$” (neglecting dependence on the constants $k,ℓ,ε$), for both RW and BFS. Despite some partial success (see below), this claim has not been proven for any $k$ and $ℓ$, neither for RW nor for BFS, with the exception of $(k,ℓ) = (2,1)$, which behaves very differently compared to other cases. A theoretical explanation for the good performance of cuckoo hashing in practical applications is therefore seriously lacking in this aspect. While this paper does not solve the problem, it puts a new kind of dent into it.

\paragraph{Contribution.} Like most previous papers on cuckoo hashing insertions (an exception \cite{FP:RandomWalk-Buckets:2018} is mentioned below) we focus on the case $ℓ = 1$. Our analysis shows that RW insertions take $𝒪(1)$ expected amortised time for all $k ≥ 3$, but it only works for $\frac{m}{n} < \cpk - ε$ where $\cpkl < \ckl$ is known as the \emph{peeling threshold} or threshold for the occurrence of an $(ℓ+1)$-core in a random $k$-uniform hypergraph \cite{PSW:SuddenCore:96,Molloy05:Cores-in-random-hypergraphs,C:Cores:2004,Luczak:A-simple-solution}, see \cref{tab:thresholds}. Our analysis extends to a setting where the $m$ insertions are carried out by $m$ threads in parallel, each executing RW. We consider the worst case, where a (possibly adversarial) scheduler arbitrarily assigns the available computation time to threads that have not yet terminated. We only assume that the scheduler is oblivious of future random choices and that swaps are atomic, i.e.\ when several threads perform swaps concurrently, the effect is the same as executing these swaps in \emph{some} sequential order (see e.g.\ \cite[Sec. 2.4]{DDMS:Basic-Toolbox:2019} for common parallel programming models).
\begin{theorem}
    \label{thm:main}
    Let $k ∈ ℕ$ with $k ≥ 3$ and $ε > 0$ be constants and $n,m ∈ ℕ$ with $\frac{m}{n} < \cpk - ε$.
    \begin{enumerate}[(i)]
      • Conditioned on a high probability event, sequentially inserting $m$ keys into a cuckoo hash table with $n$ buckets of size $1$ using RW takes $𝒪(n)$ steps in expectation.
      • The same applies if the $m$ insertions are started in parallel with arbitrary scheduling, only assuming that swaps are atomic. In other words, the combined work is $𝒪(n)$.
    \end{enumerate}
\end{theorem}

\paragraph{Related Work and Comparison.}
\cref{tab:relatedWork} summarises related work that we now discuss from left to right.
\begin{table}[b]
  \begin{tabular}{lccccc}
      \toprule
                      & \cite{Khosla:Balls-Into-Bins:2013} & \cite{FPSS:Space_Efficient:2005} & \cite{FMM:An_Analysis:2011,FPS:On_the_Insertion:2013} & \cite{FJ:Insertion-time-Cuckoo-journal:2019} & \bfseries new\\
      \midrule
          algorithm & \begin{minipage}{2cm}\footnotesize \centering offline\\construction\end{minipage} & BFS & RW & RW${}'$ & RW\\
          (expected amortised) insertion time & $𝒪(1)$ & $𝒪(1)$ & $\log^{𝒪(1)}(n)$ & $𝒪(1)$ & $𝒪(1)$\\
          least $k$ for load factor $> 50\%$ & $3$ & $≥10$ & $3$ & $≥12$ & $3$\\
          supports load factor $1-ε$ for large $k$ & ✓ & ✓ & ✓ & ✓ & ✗\\
          supports load factor $\ck-ε$ & ✓ & ✗ & ✓ & ✗ & ✗\\
          supports deletions & - & ✓ & ✓ & ✗ & ✗\\
      \bottomrule
  \end{tabular}
  \caption{Guarantees offered by analyses on cuckoo table insertions. The motivation for the third line is that any load factor $< 50\%$ can be achieved with $k = 2$.}
  \label{tab:relatedWork}
\end{table}
\begin{description}
    •[\cite{Khosla:Balls-Into-Bins:2013}] is only included here to show that the case of \emph{static} cuckoo hash tables is well understood with optimal results in all considered categories.
    •[\cite{FPSS:Space_Efficient:2005}] offers a strong analysis of BFS. The only downside is that it does not work for some small $k$ and does not reach all the way to $\ck$. These issues might be resolvable by modernising the proof (the values $\ck$ were not known at the time of writing). Note, however, that even full success on this front would not render an analysis of RW irrelevant as several authors, including \cite{FPSS:Space_Efficient:2005}, see significant \emph{practical} benefits of RW over BFS.
    •[\cite{FMM:An_Analysis:2011,FPS:On_the_Insertion:2013}] propose and improve, respectively, an analysis of RW via graph expansion. It guarantees most desired properties, including a concentration bound on insertion times. The major downside is an only polylogarithmic bound on expected insertion time.
    •[\cite{FJ:Insertion-time-Cuckoo-journal:2019}] are first to prove an $𝒪(1)$ bound on expected random walk insertion time. The proof extends to any load factor $1-ε$ with $ε > 0$. There is a downside, however. Instead of using $k = 𝒪(\log(1/ε))$ hash functions as would be required for the existence of a placement of all keys (and as are used by \cite{FPSS:Space_Efficient:2005} in their BFS analysis), the authors use $k = 𝒪(\log(1/ε)/ε)$ hash functions. To give an example, while $k = 3$ hash functions suffice for $ε = 0.2$ (because $80\% < \c_{3,1} ≈ 92\%$), the analysis of \cite{FJ:Insertion-time-Cuckoo-journal:2019} requires $k ≥ 50$ hash functions for that $ε$.
    Even if the analysis can be tuned to some degree (which seems probable), useful guarantees for practically relevant $k$ would likely remain out of reach. We remark that they use a variant RW${}'$ of RW where a key $x$ searches all its buckets $h₁(x),…,h_k(x)$ for a free space and only moves to a random bucket if all are full.
    •[\cite{FP:RandomWalk-Buckets:2018},] not shown in the table, considers $k = 2$ hash functions, buckets of size $ℓ ≥ 2$ and random walk insertion. The result resembles \cite{FJ:Insertion-time-Cuckoo-journal:2019} in its merits and downsides: Expected insertion time of $𝒪(1)$ is supported at any load factor $1-ε$, provided that $ℓ$ is large enough; the value $ℓ = ℓ(ε)$ required in the analysis is exponentially larger than what is needed for the existence of a placement; and meaningful guarantees for small values of $ℓ$ seem out of reach.
    •[This paper] is the first to guarantee constant time insertions into cuckoo hash tables using $k ∈ \{3,…,9\}$ hash functions. Like \cite{FJ:Insertion-time-Cuckoo-journal:2019} our proof does not consider deletions. The main downside is that our analysis only works for $\frac{m}{n} < \cpk - ε$. Paradoxically, this means that the load factor supported by our analysis decreases when more hash functions are used (indeed, $\cpk → 0$ for $k → ∞$) and the supremum of supported load factors is $\cp_{3,1} ≈ 0.818$ for $k = 3$ hash functions. We mention a potential avenue for overcoming this problem in the conclusion.
\end{description}

\paragraph{Technical Overview.}
A central idea in our approach is to not count the number of evictions caused by a single insertion operation but to take the perspective of a single key and count how often it moves in the course of all $m$ insertion operations combined.

Our proof is inspired by a simple observation: If a key $x$ is assigned a bucket $h_i(x)$ that is assigned to no other key, then $x$ is safely out of the way of other keys as soon as it has been placed in $h_i(x)$. In expectation, this happens after $x$ has moved $k$ times. A constant fraction of keys is “harmless” in this way. Moreover, there are keys $y$ that are assigned a bucket $h_i(y)$ that is assigned to no other key, \emph{except} for some harmless keys. It seems plausible that $y$, too, can quickly find a home in $h_i(y)$ and is only expected to be evicted from it a few times until the harmless keys live up to their name.

A formalisation attempt goes like this: Let $F$ be an injective placement function assigning to each key $x$ a bucket $F(x) ∈ \{h₁(x),…,h_k(x)\}$ (such an $F$ exists whp for $c < \ckl - ε$). We say a key $x$ \emph{depends} on a key $y$ if $F(x) ∈ \{h₁(y),…,h_k(y)\}$, i.e.\ if the position designated for $x$ is admissible for $y$. Let $D(x)$ be the set of all keys that $x$ depends on. Finally, let $\moves(x)$ be the total number of times that $x$ moves during the insertion of all keys. We then have
\begin{equation}
  𝔼[\moves(x)] ≤ k + \sum_{y ∈ D(x)} 𝔼[\moves(y)].\label{eq:overview-moves}
\end{equation}
The “$k$” is due to $x$ being first placed in $F(x)$ after $k$ moves in expectation. It can then only be evicted from $F(x)$ by a key from $D(x)$. Each movement of a key in $D(x)$ has a chance of $\frac 1k$ to evict $x$ from $F(x)$, causing $k$ more moves of $x$ in expectation until $x$ is back in $F(x)$. Hence each move of a key from $D(x)$ can cause at most one move of $x$ in expectation, as (\ref{eq:overview-moves}) suggests. The claim is even true in a more general context we call the \emph{random eviction process} where in each round an adversary choses the key $y$ to be moved among all keys not currently placed in their designated location $F(y)$.
 
 As a way of bounding $𝔼[\moves(x)]$, \cref{eq:overview-moves} is hopelessly circular at first, but it is useful for specific $F$. Indeed, assume that the configuration of keys and buckets is \emph{peelable}, i.e.\ for every subset $X'$ of keys there is a bucket $b^*$ assigned to only one key $x^*∈ X'$. In that case, we can iteratively construct $F$, always picking such a pair $(x^*,b^*)$ uniformly at random, setting $F(x^*) = b^*$ and removing $x^*$ from further consideration. This yields an \emph{acyclic} dependence relation and an acyclic \emph{dependence graph} (the directed graph with one vertex for each key that has the dependence relation as its edge relation). We can then upper bound $𝔼[\moves(x)]$ by a multiple of $\peel_F(x)$, which is the number of paths in the dependence graph that start at $x$. Bounding the expected number of moves of all insertion operations combined by $𝒪(n)$ then amounts to bounding the total number of paths in the dependence graph by $𝒪(n)$.

The second part of our argument
\iffullversion
– contained in the full version of this paper –
\fi
is intimately related to the analysis of $2$-cores in random hypergraphs. We extend Janson and Luczak's “simple solution to the $(ℓ+1)$-core problem” \cite{Luczak:A-simple-solution}, which uses a random process embedded in continuous time where peeling is applied to the configuration model of a random hypergraph. We establish two guarantees concerning the peeling process.
\begin{itemize}
  • Firstly, the guarantee that during “early” rounds of peeling (when $Ω(n)$ keys still remain) there are always $Ω(n)$ candidate pairs $(x^*,b^*)$ to choose from. Intuitively, this large number of choices for the peeling process makes it likely that the dependence graph becomes very “wide” with few long paths. For illustration (the formal argument works differently) assume the maximum path length is $w$  with $w = 𝒪(1)$. Since the indegree of the dependence graph is bounded by $k-1$ this gives a bound of $m·(k-1)^w = Θ(m)$ on the total number of paths as desired.
  • Secondly, the technically demanding guarantee that in the “late” phase of peeling (when only $o(n)$ keys remain) almost all buckets have at most one remaining key assigned to them. Most steps of the peeling process will then not create further edges in the dependence graph. This implies that for each of the paths that already exist in the dependence graph less than one additional path is created in future rounds.
\end{itemize}

\paragraph{Outline.} The rest of this paper is devoted to proving \cref{thm:main}. We introduce two notable auxiliary concepts we call the \emph{random eviction process} (REP) and \emph{peeling numbers}. We reduce \cref{thm:main} to a claim about REP (\cref{sec:REP}) and reduce this claim to an upper bound on peeling numbers (\cref{sec:peeling-numbers}).
\iffullversion
After a deep dive into hypergraph peeling (\cref{sec:peeling-in-config-model}) we obtain said upper bound by counting paths in the dependence graph (\cref{sec:bounding-pn}).
\else
The remaining technical content is found in the full version of this paper.
A deep dive into hypergraph peeling is required (full version Section 4). We then establish the upper bound on peeling numbers by counting paths in the dependence graph (full version Section 5).
\fi

\section{Orientations, Peeling and the Random Eviction Process}
\label{sec:REP}
In this section, we introduce the \emph{random eviction process} (REP), which generalises sequential RW insertions, and formulate a claim on \REPT that implies \cref{thm:main}.

\paragraph{From Hashing to Hypergraphs.} A well-subscribed model for cuckoo hashing involves hypergraph terminology. The set of buckets corresponds to the set $V$ of vertices and each key $x$ corresponds to the hyperedge $\{h₁(x),…,h_k(x)\}$ in the set $E$ of hyperedges. The task of placing all keys then becomes the task of \emph{orienting} $H = (V,E)$ as explained below.

Under the \emph{simple uniform hashing assumption}, the distribution of $H = H_{n,m,k}$ is simple: Each of the $km$ incidences of the $m$ hyperedges are chosen independently and uniformly at random from $V$. Formally this means that hyperedges are multisets of size $k$, possibly containing multiple copies of the same vertex (though in expectation only $𝒪(1)$ do) and $E$ is a multiset possibly containing identical hyperedges (though whp $E$ does not). This issue complicates a few definitions but does not cause any real trouble.

\paragraph{Orientations and Peelings.} A \emph{partial orientation} of a hypergraph $H = (V,E)$ is a function $f : E → V ∪ \{⊥\}$ with $f(e) ∈ e ∪ \{⊥\}$ for each $e ∈ E$ that is injective except for collisions on $⊥$. If $f(e) = ⊥$ then we say that $e$ is \emph{unoriented}, otherwise $e$ is \emph{oriented} (to $f(e)$). We call $f$ an \emph{orientation} if all $e ∈ E$ are oriented.

We can try to construct an orientation $F$ of $H$ greedily by repeatedly selecting a vertex $v$ of degree $1$ arbitrarily as long as one such vertex exists, setting $F(e) = v$ for the unique hyperedge $e$ incident to $v$, and removing $e$ from $H$.
We call the resulting partial orientation $F$ a \emph{peeling} of $H$. If $H$ does not contain a subhypergraph of minimum degree at least $2$ (i.e.\ when the $2$-core of $H$ is empty \cite{Molloy05:Cores-in-random-hypergraphs}) then $F$ is an orientation and we say $H$ is \emph{peelable}. We call $F$ a \emph{random peeling} if the choice of $v$ is made uniformly at random whenever there are several vertices of degree $1$.


\paragraph{The Random Eviction Process.} The \emph{random eviction process} (REP), see \cref{algo:rep} (left), is run on a hypergraph $H = (V,E)$ and maintains a partial orientation $f$ of $H$. The process continues in a sequence of rounds as long as unoriented hyperedges remain, possibly indefinitely. In each round, an unoriented hyperedge $e$ is chosen and oriented to a random incident vertex. If a different hyperedge $e'$ was oriented to that vertex, then this $e'$ is \emph{evicted}, i.e.\ becomes unoriented.

A variant of \REP is the \emph{random eviction process with target orientation} (\REP'), see \cref{algo:rep} (right). It is run on a hypergraph $H$ and an orientation $F$ of $H$. \REPT works just like \REP, except that it terminates only when $f = F$ is reached, and in every round any hyperedge $e$ with $f(e) ≠ F(e)$ may be chosen. We claim:

\def\iold{i_{\mathrm{old}}}
\begin{figure}
  \begin{minipage}{0.48\textwidth}
    \begin{algorithm}[H]
      \algo{\normalfont REP($V,E$)}{
        $f ← \{(e,⊥) \mid e ∈ E\}$\tcp{i.e.\ $f≡⊥$}
        \While{$∃e ∈ E: f(e) = ⊥$}{
          pick such an $e$ arbitrarily\;
          pick a random $v ∈ e$\;
          \If{$∃e'≠e: f(e') = v$}{
            $f(e') ← ⊥$\;
          }
          $f(e) ← v$\;
        }
      }
    \end{algorithm}
  \end{minipage}
  \begin{minipage}{0.50\textwidth}
    \begin{algorithm}[H]
      \algo{\normalfont \bfseries \REPT{}$\bm{(V,E,F)}$}{
        $f ← \{(e,⊥) \mid e ∈ E\}$\;
        \While{$∃e ∈ E: {\bm{f(e) ≠ F(e)}}$}{
          pick such an $e$ arbitrarily\;
          pick a random $v ∈ e$\;
          \If{$∃e'≠e: f(e') = v$}{
            $f(e') ← ⊥$\;
          }
          $f(e) ← v$\;
        }
      }
    \end{algorithm}
  \end{minipage}
  \caption{The random eviction process (REP) is a generalisation of sequential random walk insertion. A variant \REPT only terminates when a specific target orientation $F : E → V$ is reached. Changes are highlighted in bold.}
  \label{algo:rep}
\end{figure}

\begin{proposition}
  \label{prop:main}
  Let $k ∈ ℕ$ with $k ≥ 3$ and $ε > 0$ be constants and $n,m ∈ ℕ$ with $\frac{m}{n} < \cpk - ε$.
  Conditioned on a high probability event, $H = H_{n,m,k}$ is peelable and the random peeling $F$ of $H$ satisfies the following.
  \REPT with target orientation $F$ and an arbitrary\footnote{This allows these choices to be made \emph{adversarially}. The adversary may know all about $H$ and the state of the algorithm but cannot predict future random choices made in line 5.} policy for choosing $e$ in line 4 terminates after $𝒪(n)$ rounds in expectation.
\end{proposition}

\noindent \Cref{prop:main} is proved in the following section. We now show that it implies \cref{thm:main}.

\begin{proof}[Proof of \cref{thm:main}.]
  
  The case of $m$ sequential insertions is equivalent to the case of $m$ parallel insertions where the scheduler only assigns computation time to the thread of least index that has not yet terminated. It therefore suffices to prove \lipicsLabel{(ii)}, where the parallel case with \emph{arbitrary} scheduling is considered.
  

  We deal with $m$ threads, each running RW, executed in an arbitrarily interleaved way. However, the only point where RW interacts with data visible to other threads is the swap, which is assumed to be atomic.
  A sufficiently general case is therefore one where the scheduler always picks an arbitrary thread that has not yet terminated and  that thread is then allowed to run for one iteration of the loop. The correspondence between this process and REP should be clear: The scheduler's arbitrary choice of a thread implicitly chooses an unplaced key in that thread's local variable $x$, which is then placed into a random bucket, possibly evicting a different key. Likewise, REP arbitrarily chooses an unoriented hyperedge, which is then randomly oriented, possibly evicting another hyperedge.

  For \cref{prop:main}'s claim on \REPT to apply to RW, there are two differences to consider.
  \begin{description}
    •[\REP vs.\ \REPT] Assume an adversary wants to \emph{maximise} the expected running times of \REP and \REPT by making bad choices for $e$ in line 4. Her job is \emph{harder} for \REP for two reasons: Firstly, the termination condition is strictly weaker, such that \REP may terminate when \REPT does not. Secondly, her choices for $e$ are restricted to unoriented hyperedges, where \REPT additionally permits oriented hyperedges with $f(e) ≠ F(e)$.

    Intuitively speaking, the relatively weaker adversary in \REP means that the upper bound on expected running time in \cref{prop:main} carries over from \REPT to \REP. More formally, every policy $P$ for line 4 of REP is also valid for REP' and under the natural coupling random coupling REP' with $P$ takes always at least as long as REP with $P$.
    •[\REP vs. RW: $\iold$.] In RW an evicted key is not allowed to immediately move back into the bucket $\iold$ it was just evicted from. The intuition is that this avoids a needless back-and-forth that otherwise occurs in $1$ out of every $k$ evictions. However, the author is not aware of a simple proof that the use of $\iold$ is an improvement. Instead, we will check that the relevant part of the argument (\cref{lem:number-of-moves}) works for both cases.\qedhere
  \end{description}
\end{proof}

\section{Bounding the Number of Evictions using Peeling Numbers}
\label{sec:peeling-numbers}

We now introduce the concept of peeling numbers and bound the number of evictions occurring in \REPT in terms of them. This proves \cref{prop:main} but leaves the task of bounding peeling numbers for
\iffullversion
  later.
\else
 the full version of this paper.
\fi

\paragraph{Direct Dependence and Numbers of Moves.} Consider a peelable hypergraph $H = (V,E)$ and a peeling $F : E → V$ of $H$. For $e≠e' ∈ E$ we say that $e$ \emph{directly depends} on $e'$ if $F(e) ∈ e'$. This implies that $e$ is peeled after $e'$, making the transitive closure of direct dependence an acyclic relation. We define $D(e) = D_F(e)$ as the set of all $e'$ that $e$ directly depends on, or more precisely: $D(e)$ is a multiset containing $e'$ with the same multiplicity with which $e'$ contains $F(e)$.

Now consider a run of \REPT with target orientation $F$ (and an arbitrary policy for line 4). For $e ∈ E$ let $\moves(e)$ be the number of times that $e$ is selected in line 4 of \REPT (this is one more than the number of times that $e$ is evicted).

\begin{lemma}
  \label{lem:number-of-moves}
  For any $e ∈ E$ we have
    $𝔼[\moves(e)] ≤ k + \sum_{e' ∈ D(e)} 𝔼[\moves(e')]$.
\end{lemma}

\begin{proof}
  For clarity, we ignore complications that are due to multisets at first. Let $m₁$ be the number of times that $e$ moves until $f(e) = F(e)$ holds for the first time. Whenever $e$ is selected to be moved, the chance to select $f(e) = F(e)$ is $\frac{1}{k}$, so clearly $𝔼[m₁] = k$. Afterwards, $e$ may not be selected anymore until evicted. Only hyperedges in $D(e)$ can evict $e$ from $F(e)$ and when selected they do so with probability $\frac{1}{k}$, causing another $k$ moves of $e$ in expectation. 
  It follows that $𝔼[m₊] = 𝔼[m_D]$ where $m₊ := \moves(e) - m₁$ and where $m_D$ is the number of times that a hyperedge from $D(e)$ moves \emph{while} $f(e) = F(e)$. The claim now follows from $m_D ≤ \sum_{e' ∈ D(e)} \moves(e')$ and  linearity of expectation.
  
  When $D(e)$ is a multiset the argument can be adapted: Whenever a hyperedge $e'$ moves that is contained in $D(e)$ with multiplicity $a > 1$ it has an increased chance of $\frac{a}{k}$ to move to $F(e)$. But this is reflected in our bound since $𝔼[\moves(e')]$ is counted $a$ times.

  Let us now consider a variant of the claim that incorporates the “$\iold$” feature of RW as promised in the proof of \cref{thm:main}. In particular, a hyperedge never moves into the position it was last evicted from. We now have $𝔼[m₁] < k$ because all moves after the first move have an improved chance of $\frac{1}{k-1}$ to select $F(e)$. To compare $𝔼[m₊]$ and $𝔼[m_D]$, we can distinguish two kinds of moves. Concerning moves \emph{away from} $F(e)$, $m_D$ counts the same or one more compared to $m₊$. All other moves have a chance of $\frac{1}{k-1}$ to end in $F(e)$ and contribute the same amount to $𝔼[m₊]$ and $𝔼[m_D]$ as before. The same adaptation to multisets applies.
\end{proof}

\paragraph{The Peeling Number.}
We define the peeling number of $e ∈ E$ recursively as
\begin{equation}
  \peel(e) = \peel_F(e) := \sum_{e' ∈ D_F(e)} (1+\peel(e')).\label{def:pn}
\end{equation}
Peeling numbers are well-defined by acyclicity of direct dependence, the base case being $\peel(e) = 0$ for any $e$ with $D(e) = ∅$.
The idea is that $\peel(e)$ counts the number of hyperedges that $e$ directly \emph{or indirectly} depends on, in other words, those hyperedges $e'$ that must be peeled before $e$ can be peeled. However, some $e'$ may be counted multiple times.
\iffullversion
  A more precise interpretation as a number of dependence \emph{paths} is given in \cref{sec:bounding-pn}.
\fi
The relevance of peeling numbers lies in the following lemma.

\begin{lemma}
  \label{lem:dep-to-moves}
  Let $H$ be a peelable hypergraph with a peeling $F$. Let $R$ be the number of rounds until \REPT with target orientation $F$ terminates. We have $𝔼[R] ≤ k·(m+\sum_{e ∈ E}\peel_F(e))$.
\end{lemma}

\begin{proof}
  For a single $e ∈ E$ we have $𝔼[\moves(e)] ≤ k·(1+\peel(e))$ because
  \begin{align*}
    𝔼[\moves(e)] &\stackrel{Lem. \ref{lem:number-of-moves}}{≤} k + \sum_{e' ∈ D(e)} 𝔼[\moves(e')] \stackrel{\mathrm{Induction}}{≤} k + \sum_{e' ∈ D(e)} k·(1+\peel(e'))\\
    &\eqRefRel{def:pn}{Eq.}{=} k+k·\peel(e) = k·(1+\peel(e)).
  \end{align*}
  Since the total number $R$ of rounds of \REPT is the sum of all moves we conclude
  \[𝔼[R] = 𝔼\Big[\sum_{e∈E}\moves(e)\Big] 
  ≤ \sum_{e∈E}k·(1+\peel(e)) = k·(m+\sum_{e∈E}\peel(e)).\qedhere\]
\end{proof}

The remaining technical challenge is to bound the sum of all peeling numbers:

\begin{proposition}
  \label{prop:peeling-numbers-small}
  Let $H$ be as in \cref{prop:main}. There is a high probability event $\mathcal{E}$ such that, conditioned on $\mathcal{E}$, $H$ is peelable and the peeling numbers with respect to the random peeling $F$ of $H$ satisfy
  \[𝔼\Big[\sum_{e ∈ E} \peel_F(e) \mid \mathcal{E}\Big] = 𝒪(n).\]
\end{proposition}
\iffullversion
  Let us show how \cref{prop:peeling-numbers-small} implies \cref{prop:main}.
\else
  A prove is found in the full version of this paper (Section 5) and requires a detailed analysis of the peeling process (Section 4). We conclude this extended abstract with showing how \cref{prop:peeling-numbers-small} implies \cref{prop:main}.
\fi

\begin{proof}[Proof of \cref{prop:main}]
  We take the opportunity to clarify the structure of our probability space. There are three random experiments, performed in sequence: First, we pick a random hypergraph $H$. Second, if $H$ is peelable, we pick a random peeling $F$ of $H$ and observe the peeling numbers. Last, we execute \REPT{}$(H,F)$ and observe which moves are made. Note that the high probability event $\mathcal{E}$ from \cref{prop:peeling-numbers-small} only relates to the first two steps (it does not relate to any moves). \Cref{lem:dep-to-moves} only relates to the last step and does not require $H$ and $F$ to be random.
  For the number $R$ of rounds of $\REPT$ we obtain:
  \begin{align*}
    𝔼[R \mid \mathcal{E}] &= 𝔼\big[𝔼[R \mid H,F,\mathcal{E}] \mid \mathcal{E}\big] = 𝔼\big[𝔼[R \mid H,F] \mid \mathcal{E}\big]
    \refrel{lem:dep-to-moves}{Lem.}{≤} 𝔼\Big[k·(m+\sum_{e∈E}\peel(e)) \mid \mathcal{E}\Big]\\
    &= km + k·𝔼\Big[\sum_{e∈E}\peel(e) \mid \mathcal{E}\Big]
    \refrel{prop:peeling-numbers-small}{Prop.}{=} km + k·𝒪(n) = 𝒪(n).\qedhere
  \end{align*}
\end{proof}

\iffullversion

\section{Peeling in the Configuration Model in Continuous Time}
\label{sec:peeling-in-config-model}

Following ideas by Janson and Luczak \cite{Luczak:A-simple-solution} we introduce two shifts in perspective that affect the structure of the probability space in which $H = H_{n,m,k}$ and the random peeling $F$ of $H$ reside. The first is to use a \emph{configuration model}, which allows jointly revealing $F$ and $H$ in small increments, instead of “wasting” the randomness contained in $H$ in the beginning. The second is a switch to \emph{continuous time} that makes some calculations more elegant.

Let us try to point out which of this section's content is genuinely new. Janson and Luczak \cite{Luczak:A-simple-solution} were interested in $(ℓ+1)$-cores of random graphs. Since our peeling process relates to $2$-cores of $k$-uniform hypergraphs, we have to generalise their work to $k > 2$ while collapsing $ℓ$ to $1$. Moreover, we require a more detailed understanding of the “late” rounds of peeling. Ultimately we could not reuse any existing results due to formal complications. If pressed we would claim that \cref{prop:ball-counts} \itemRef{it:many-light-balls} was essentially known and \cref{prop:ball-counts} \itemRef{it:few-heavy-balls} is new.

\paragraph{Peeling in the Configuration Model.}
The configuration model is a convenient way to represent distributions on graphs, see e.g.\ \cite{C:Cores:2004,FR:The_k-orientability:2007,MR:GiantComponentDegreeSequence:1998,Luczak:A-simple-solution,K:Poisson:2006}. 
The probability space is given by a sequence $(d₁,…,d_n)$ of degrees where $2m = \sum_{i = 1}^n d_i$ is even. A random (multi-) graph is obtained from the distribution by attaching a corresponding set of $2m$ half-edges to $n$ vertices and randomly matching the half-edges to form edges.

Applying this to hypergraphs allows us to reveal $H = H_{n,m,k}$ in two separate steps. First, $km$ balls are independently and uniformly assigned to the $n$ vertices. The numbers $d₁,…,d_n$ of balls at each vertex form the degree sequence of $H$. They also constitute a configuration model for the second step. Indeed, the remaining randomness lies in partitioning the $km$ balls uniformly at random into $m$ sets of size $k$ each. Each set of $k$ balls corresponds to one hyperedge of $H$ given by the corresponding multiset of $k$ vertices. It is useful to realise that the partitioning of the balls need not be revealed all at once but can be revealed in tandem with a random peeling $F$ of $H$.


Concretely, in each round of the peeling process, we select a random vertex $v$ of degree $1$ and the single ball at that vertex. We then reveal the hyperedge that this ball contributes to by selecting $k-1$ other balls uniformly at random. This reveals one hyperedge $e$ of $H$ and $F(e) = v$. After removing the $k$ balls we obtain a configuration model for the remaining hyperedges. We continue the process as long as at least one vertex of degree $1$ remains.

\paragraph{Continuous Time.}
We now introduce an equivalent way of selecting remaining balls uniformly at random. Each ball is independently assigned a random lifetime with distribution $\Exp(1)$. With probability $1$ all lifetimes are different. The peeling process is now embedded in continuous time. We begin at $t = 0$ and repeat the following two steps. First, select a vertex of degree $1$ uniformly at random and instantly remove its ball (i.e.\ this ball is removed before its lifetime has expired). Then, wait for $k-1$ balls to die (these are also removed) and form a hyperedge with these $k$ balls as before. The process continues until the \emph{stopping time} $τ$ when no vertex of degree $1$ can be selected any more.

\paragraph{Light and Heavy Balls.}
Let $B(t)$ for $0 ≤ t ≤ τ$ be the number of remaining balls at time $t$ (excluding balls removed at time $t$).
We call a ball \emph{light} if it is the only remaining ball in its vertex and \emph{heavy} otherwise. Let $L(t)$ and $H(t)$ with $L(t)+H(t) = B(t)$ be the corresponding ball counts at different points $0 ≤ t ≤ τ$ in time.

The following proposition is a central ingredient for proving \cref{prop:peeling-numbers-small} in \cref{sec:bounding-pn}. It claims that \lipicsLabel{(i)} in the short term there is always a linear number of light balls and \lipicsLabel{(ii)} in the long run, the heavy balls make up only a small fraction of the remaining balls.
\begin{proposition}
  \label{prop:ball-counts}
  \begin{enumerate}[(i)]
    • \label{it:many-light-balls} Let $c = \cpk - ε$. Then for any $t₀ ≥ 0$ we have $τ ≥ t₀$ whp and there exists $δ = δ(ε,t₀) > 0$ such that
    $\min_{t ∈ [0,t₀]} L(t) ≥ δn$ whp.
    • \label{it:few-heavy-balls} There exists a constant $t₀$ such that $H(t) ≤ \frac{1}{2k}B(t)$ for all $t ∈ [t₀,τ]$ whp.
  \end{enumerate}
\end{proposition}
The rest of this section is devoted to proving \cref{prop:ball-counts} and is quite technical. The author would recommend skipping to \cref{sec:bounding-pn} on a first read.



\subsection{The expectation of \texorpdfstring{$\bm{B(t)}$}{B(t)} and \texorpdfstring{$\bm{H(t)}$}{H(t)}}

In this subsection we rectify the annoying fact that $B(t)$ and $H(t)$ are not well-defined for $t > τ$, we link $B(t)$ to a pure death process and we determine the expectation of $B(t)$ and $H(t)$ up to constant additive errors.

\paragraph{Extending $\bm{B(t)}$ and $\bm{H(t)}$ beyond $\bm{τ}$.}
When reasoning about $B(t)$ and $H(t)$ individually, we may use the following definitions.
\begin{description}
    •[$B'(t)$:] We begin with $km$ balls with $\Exp(1)$-distributed lifetimes and repeat the following two steps until no balls remain: \lipicsLabel{(i)} remove a ball uniformly at random \lipicsLabel{(ii)} wait for $k-1$ balls to die (and remove them). Define $B'(t)$ as the number of balls remaining at time $t$.\\
    Note: No reference to vertices is made. In contrast, the peeling process insisted that the ball removed in \lipicsLabel{(i)} is light.
    •[$H'(t)$:] We have $km$ balls that are randomly assigned to $n$ vertices, that have $\Exp(1)$ distributed lifetimes, and that are removed on death.
    Define $H'(t)$ to be the number of balls that are not alone at their vertex at time $t$.
    \\
    Note: The peeling process insisted that a light ball is removed in each of its rounds. Here, we are indifferent about the removal of such lonely balls as these are not counted in $H'(t)$ anyway.
\end{description}
The peeling process constitutes an obvious three-way coupling between $(B(t),H(t))$, $B'(t)$ and $H'(t)$ with $B(t) = B'(t)$ and $H(t) = H'(t)$ for all $t ∈ [0,τ]$. Now, however, consider the case where the peeling process terminates early with $B(τ) = H(τ) > 0$ because there is no light ball to remove. The processes $B'(t)$ and $H'(t)$ have no such issue and continue naturally until $B'(t) = 0$ and $H'(t) = 0$, respectively.

In the following we assume that $B(t)$ and $H(t)$ refer to the extended processes $B'(t)$ and $H'(t)$, which are coupled to the peeling process for $t ≤ τ$. While the values $B(t)$ and $H(t)$ are \emph{meaningful} only for $t ≤ τ$, they are now \emph{well-defined} for all $t ∈ ℝ_{≥0}$.

\begin{lemma}[Purifying $B(t)$.]
  \label{lem:coupling-of-Bt}
  Let $P(t)$ be a \emph{pure death process} with initially $(k-1)m$ elements and parameter $\frac{k}{k-1}$, i.e.\ elements have independently distributed lifetimes $X₁,…,X_{(k-1)m} \sim \Exp(\frac{k}{k-1})$ and $P(t) := |\{i ∈ [(k-1)m] \mid X_i > t\}|$.

  There is a coupling between $B(t)$, $P(t)$ and a copy $P'(t)$ of $P(t)$ such that
  \[\tfrac{k}{k-1}P(t)-1 ≤ B(t) ≤ \tfrac{k}{k-1}P'(t)+1 \text{ for any $t ≥ 0$}.\]
\end{lemma}

\begin{proof}
  Let $t₁ < t₂ < … < t_{(k-1)m}$ be the times when $P(t)$ jumps and $t₀ = 0$. A death process jumps with a rate proportional to the number of remaining elements, i.e.\ the waiting times $Δ_i := t_i-t_{i-1}$ for $i ∈ [(k-1)m]$ are independent with distribution $Δ_i \sim \Exp(λ_i)$ where $λ_i = \frac{k}{k-1}((k-1)m-i+1)$.

  Similarly, let $0 = t₀' < t₁' < … < t'_{(k-1)m}$ be the times when $B(t)$ jumps. Since every round of the peeling process begins and ends with the removal of a ball, $B(t)$ jumps two steps at $t_i'$ when $i ≡ 0 \pmod{k-1}$, except for $i = 0$ and $i = (k-1)m$ where it jumps by just $1$ step. Defining $Δ_i' := t'_i - t'_{i-1}$ for $i ∈ [(k-1)m]$ gives $Δ_i' \sim \Exp(λ_i')$ with $λ_i' = km-i- ⌊\frac{i-1}{k-1}⌋$. It is easy to check that $λ_i ≥ λ_i'$. In particular we can couple $Δ_i$ and $Δ_i'$ (i.e.\ embed them in a joint probability space) such that $Δ_i ≤ Δ_i'$ for all $i ∈ [(k-1)m]$. This implies $t_i ≤ t_i'$ for all $i ∈ [(k-1)m]$. Using monotonicity of $P(t)$ and the fact that $P(t_i)$ and $B(t_i')$ are fixed we can derive
  \begin{align*}
    \tfrac{k}{k-1}P(t_i') &≤ \tfrac{k}{k-1}P(t_i) = \tfrac{k}{k-1}((k-1)m - i) = km - i - \tfrac{i}{k-1}\\
    &≤ km-i-⌊\tfrac{i}{k-1}⌋ = B(t_i')+1. 
  \end{align*}
  Since this covers all time points when $B(t)$ jumps, $\frac{k}{k-1}P(t)-1 ≤ B(t)$ follows from monotonicity of $B(t)$ and $P(t)$. A similar coupling between $B(t)$ and $P(t)$ (this time coupling $Δ_{i+1}$ and $Δ_i'$ such that $Δ_{i+1} ≥ Δ_i'$) yields the second half of the claim.
\end{proof}

\begin{lemma}[Expectation of $B(t)$ and $H(t)$.]\ \\
  \label{lem:expectation-B-H}
  Let $B^*(p) = ckp^{k/(k-1)}$ and $H^*(p) := ckp(1-\e^{-ckp})$. Then:
  \begin{enumerate}[(i)]
        • \label{it:EB}  $∀t ≥ 0: 𝔼[B(t)]\, = n·B^*(\e^{-t}) \pm 𝒪(1)$.
        • \label{it:EH}  $∀t ≥ 0: 𝔼[H(t)]   = n·H^*(\e^{-t}) \pm 𝒪(1)$.
  \end{enumerate}
\end{lemma}

\begin{proof}
  \begin{enumerate}[(i)]
    • Using \cref{lem:coupling-of-Bt} and the definition of the pure death process $P(t)$ yields
    \begin{align*}
        𝔼[B(t)] &= 𝔼[\tfrac{k}{k-1}P(t)] \pm 𝒪(1) = \tfrac{k}{k-1}(k-1)m·\Pr[\Exp(\tfrac{k}{k-1}) > t] \pm 𝒪(1)\\
        &= km\e^{-t·\frac{k}{k-1}} \pm 𝒪(1) = n·ck(\e^{-t})^{\frac{k}{k-1}} = n·B^*(\e^{-t}) \pm 𝒪(1).
    \end{align*}
    • Consider any fixed vertex $v$ and the number $X$ of balls at $v$ with lifetime exceeding $t$. The contribution of the balls at $v$ to $H(t)$ is $X·𝟙_{X ≥ 2}$ and the distribution of $X$ is $\Bin(km,\frac{1}{n}\e^{-t})$ since each of the $km$ balls is assigned to $v$ with probability $\frac{1}{n}$ and has lifetime $> t$ with probability $\e^{-t}$. In the following we use $(1-\frac 1x)^x = \frac 1\e \pm 𝒪(\frac{1}{x})$, which can be derived from $(1+\frac{1}{x})^x < \e < (1+\frac{1}{x})^{x+1}$.
    \begin{align*}
        𝔼[H(t)] &= n·𝔼[X·𝟙_{X ≥ 2}] = n(𝔼[X]-\Pr[X = 1])\\
        &= n·\big(km·\tfrac{1}{n} \e^{-t}-km·\tfrac{1}{n} \e^{-t} · (1-\tfrac{1}{n}\e^{-t})^{km-1}\big)\\
        &= n·ck\e^{-t}\big(1-(1-\tfrac{1}{n\e^{t}})^{n\e^{t}·\frac{km-1}{n\e^{t}}}\big)
        = n·ck\e^{-t}\big(1-(1/e \pm 𝒪(\tfrac{1}{n\e^{t}}) )^{\frac{km-1}{n\e^{t}}}\big)\\
        &= n·ck\e^{-t}\big(1-\e^{-ck\e^{-t}}\big)\pm 𝒪(1) = n·H^*(\e^{-t}) \pm 𝒪(1)\qedhere
    \end{align*}
  \end{enumerate}
\end{proof}

\subsection{Concentration bounds for the early rounds of peeling}
\label{sec:early-rounds-new}

We now apply McDiarmid's inequality (restated in \cref{lem:std-conc}) to show that $B(t) ≈ n·B^*(\e^{-t})$ and $H(t) ≈ n·H^*(\e^{-t})$ whp (\cref{lem:concentration-early}). Below the peeling threshold we have $B^*(p) > H^*(p)$ for all $p ∈ (0,1]$ (\cref{lem:heavy-balls-in-early-rounds}), which then implies that $L(t) = B(t)-H(t) = Ω(n)$ for all $t ∈ [0,t₀]$ whp when $t₀$ is a constant. This establishes \cref{prop:ball-counts} \itemRef{it:many-light-balls}.

\begin{lemma}[Hoeffding + McDiarmid]
    \label{lem:std-conc}
    The following bound \itemRef{it:hoeffding} is known as Hoeffding's inequality and its generalisation \itemRef{it:mcdiarmid} as McDiarmid's inequality \cite{D:McDiarmid:1989}.
    \begin{enumerate}[(i)]
        • \label{it:hoeffding} Let $X₁,…,X_n$ be independent Bernoulli random variables. Then
        \[ \textstyle \Pr\Big[ \big|\sum_i X_i - 𝔼\big[\sum_i X_i\big]\big| ≥ s\Big] = 2\exp(-2s²/n).\]
        • \label{it:mcdiarmid} Let $c ∈ ℝ^+$ and $f$ a real-valued function with $n$ parameters such that changing any single parameter can change the function value by at most $c$. Moreover, let $X₁,…,X_n$ be independent random variables that are valid parameters for $f$. Then
        \[ \Pr\Big[ \big|f(X₁,…,X_n) - 𝔼[f(X₁,…,X_n)]\big| ≥ s\Big] = 2\exp(-2s²/(c²n)).\]
    \end{enumerate}
\end{lemma}

\begin{lemma}
  \label{lem:concentration-early}
  With $B^*(p)$ and $H^*(p)$ as in \cref{lem:expectation-B-H} we have
  \begin{enumerate}[(i)]
      • \label{it:CB}  $∀t,a ≥ 0: \Pr\big[|B(t)-n·B^*(\e^{-t})| ≥ a \sqrt{n}\big] = \exp(-Ω(a²))$.
      • \label{it:CH}  $∀t,a ≥ 0: \Pr\big[|H(t)-n·H^*(\e^{-t})| ≥ a \sqrt{n}\big] = \exp(-Ω(a²))$.
      • \label{it:CB2} $\sup_{t ≥ 0} |B(t)-n·B^*(\e^{-t})| = o(n)$ whp.
      • \label{it:CH2} $\sup_{t ≥ 0} |H(t)-n·H^*(\e^{-t})| = o(n)$ whp.
  \end{enumerate}
\end{lemma}
\begin{proof}\begin{enumerate}[(i)]
    • By \cref{lem:expectation-B-H} \itemRef{it:EB} and \cref{lem:coupling-of-Bt} it suffices that we show $\Pr\big[|P(t)-𝔼[P(t)]| ≥ a \sqrt{n}\big] = \exp(-Ω(a²))$. Since $P(t) \sim \Bin((k-1)m,\e^{-t\frac{k}{k-1}})$ is binomially distributed and hence the sum of $(k-1)m = Θ(n)$ independent Bernoulli random variables we can use Hoeffding's inequality (\cref{lem:std-conc} \itemRef{it:hoeffding}) with $s = a\sqrt{n}$.
    • Let $X_i$ for $i ∈ [km]$ be the pair of the vertex and the lifetime assigned to the $i$-th ball. Note that $H(t)$ can be computed from $X₁,…,X_{km}$. Changing a single $X_i$ can affect $H(t)$ by at most $\pm 2$: It can cause a combination of the effects $\{$ ball $i$ now contributes to $H(t)$ and did not before, ball $i$ no longer contributes to $H(t)$ and used to before, a single ball that used to be co-located with ball $i$ at time $t$ becomes light, a single ball that is now co-located with ball $i$ at time $t$ becomes heavy $\}$.
    Since $X₁,…,X_{km}$ are independent we can use McDiarmid's inequality (\cref{lem:std-conc} \itemRef{it:mcdiarmid}) to get
    \begin{align*}
      &\Pr\Big[\big|H(t) - n·H^*(e^{-t})\big| ≥ a\sqrt{n}\Big]\\
      \stackrel{\text{Lem. \ref{lem:expectation-B-H} \itemRef{it:EH}}}{≤}
      &\Pr\Big[\big|H(t) - 𝔼[H(t)]\big| ≥ a\sqrt{n}- 𝒪(1)\Big] \stackrel{\textrm{Mc.D.}}{=} \exp(-Ω(a²)).
    \end{align*}
    • Let $0 ≤ t₁ ≤ … ≤ t_j$ be all $j = 𝒪(\sqrt{n})$ points in time $t$ where $n·B^*(\e^{-t})$ is an integer multiple of $\sqrt{n}$. For each $t_i$ individually we have $\Pr\big[|B(t_i)-n·B^*(\e^{-t_i})| ≥ n^{2/3}\big] = \exp(-Ω(n^{1/3}))$ by \itemRef{it:CB}. By union bound we have $|B(t_i)-n·B^*(\e^{-t_i})| ≤ n^{2/3}$ \emph{simultaneously} for all $i ∈ [j]$ with probability $1-\exp(-Ω(n^{1/3}))·𝒪(\sqrt{n})$, hence whp. By choice of the time points and by monotonicity of $B(t)$ and $B^*(e^{-t})$ we have $|B(t)-n·B^*(\e^{-t})| ≤ n^{2/3} + \sqrt{n}$ simultaneously for all $t ∈ ℝ_{≥ 0}$. This implies the claim.
    • We can argue analogously to \itemRef{it:CB2}.\qedhere
\end{enumerate}\end{proof}

\begin{lemma}
  \label{lem:heavy-balls-in-early-rounds}
    Let $c = \cpk - ε$. Then for any $t₀ > 0$ there exists $δ = δ(ε,t₀) > 0$ such that
    \[
      \inf_{t ∈ [0,t₀]} (B^*(\e^{-t}) - H^*(\e^{-t})) ≥ 2δ.\]
\end{lemma}

\begin{proof}
    Consider solutions to $H^*(p) = B^*(p)$. Intuitively, they indicate points in time when the peeling process cannot continue because all remaining balls are heavy. We ignore the \emph{trivial solution} $p = 0$ when $H^*(p) = B^*(p) = 0$ (reached on successful termination). Using the substitution $λ = ckp$ we get
    \begin{align*}
        H^*(p) = B^*(p) &⇔ ckp(1-\e^{-ckp}) = ckp^{\frac{k}{k-1}} ⇔ (1-\e^{-ckp})^{k-1} = p\\
        &⇔ (1-\e^{-λ})^{k-1} = \frac{λ}{ck} ⇔ c = \frac{λ}{k(1-\e^{-λ})^{k-1}}
    \end{align*}
    Thus $\cpk := \min_{λ ≥ 0} \frac{λ}{k(1-\e^{-λ})^{k-1}}$ is the smallest value of $c$ for which a non-trivial solution exists. We have thereby recovered the formula for the peeling threshold (see \cite{Molloy05:Cores-in-random-hypergraphs} for the original reference and \cite[Section 2]{DGMMPR:Tight:2010} for a more accessible explanation).
    
    The claim is now easily proved. For $c = \cpk-ε$ we have, by definition, only the trivial solution $p = 0$ and hence $B^*(\e^{-t}) ≠ H^*(\e^{-t})$ for any $t ≥ 0$. The claim follows from $B^*(1) > H^*(1)$, continuity of $B^*$ and $H^*$ and compactness of $[0,t₀]$.
\end{proof}
  
\noindent We now show that throughout all early rounds of peeling there are $Ω(n)$ light balls.
\begin{proof}[Proof of \cref{prop:ball-counts} \itemRef{it:many-light-balls}]
  Let $c = \cpk - ε$, let $t₀ ≥ 0$ and let $δ = δ(ε,t₀)$ be the constant from \cref{lem:heavy-balls-in-early-rounds}. Whp we have
    \def\range{\!{0≤t≤t₀}\!}
    \begin{align*}
      &\min_{\range} L(t)
      \stackrel{\text{def}}{=}
        \min_{\range} B(t)-H(t)\\
      ≥ &\min_{\range} \Big(
          \big(B(t) - nB^*(\e^{-t})\big)
          - \big(H(t) - nH^*(\e^{-t})\big)
          + n\big(B^*(\e^{-t}) - H^*(\e^{-t})\big)
        \Big)\\
      ≥ &
         -\underbrace{\sup_{\range} |B(t) - nB^*(\e^{-t})|}%
         _{o(n) \text{ by \cref{lem:concentration-early} \itemRef{it:CB2}}}
         -\underbrace{\sup_{\range} |H(t) - nH^*(\e^{-t})|}%
         _{o(n) \text{ by \cref{lem:concentration-early} \itemRef{it:CH2}}}
          + n· \underbrace{\inf_{\range} \big(B^*(\e^{-t}) - H^*(\e^{-t})\big)}%
         _{≥2δ \text{ by \cref{lem:heavy-balls-in-early-rounds} }}\\
      ≥ &-o(n)-o(n)+2δn ≥ δn.
    \end{align*}
    Since $L(τ) = 0$ we also have $τ > t₀$ whp as claimed.
\end{proof}

\subsection{Concentration bounds for the late rounds of peeling}

Before proving \cref{prop:ball-counts} \itemRef{it:few-heavy-balls} we establish further tail bounds on $B(t)$, $H(t)$ as well as $S(t)$, where $S(t)$ is the number of balls with an \emph{assigned lifetime} exceeding $t$. Note that $B(t)$ may be less than $S(t)$ due to balls that are selected for removal before their lifetime expires.

Note that \cref{lem:concentration-early} \itemRef{it:CB2} and \itemRef{it:CH2} are more precise for small $t$ but cease to yield meaningful bounds when fewer than $\sqrt{m}$ balls remain, e.g.\ when $t > \frac 12 \log m$.

\begin{lemma}
  \label{lem:concentration-bounds}
    Let $γ > 0$ be a small constant.
    We have
    \begin{enumerate}[(i)]
        • \label{it:exp-S}
        $𝔼[S(t)] = km\e^{-t}$ for any $t ≥ 0$.
        • \label{it:concen-S}
        $S(t) = 𝒪(m\e^{-(1-γ)t}) ∩ Ω(m\e^{-(1+γ)t})$ for all $t ∈ [0,\frac{9}{10} \ln(m)]$ whp.
        • \label{it:concen-B}
        $B(t) = Ω(P(t)) = Ω(m\e^{-(\frac{k}{k-1}+γ)t})$ for all $t ∈ [0,\frac{3}{5} \ln(m)]$ whp.
        • \label{it:exp-H}
        $𝔼[H(t)] = 𝒪(m\e^{-2t})$ for $t ≥ 0$.
        • \label{it:concen-H0}
        $H(\frac{3}{5}\ln(m)) = 0$ whp.
        • \label{it:concen-H}
        $H(t) = 𝒪(m\e^{-(2-2γ)t})$ for all $t ∈ [0,\frac{3}{5}\ln(m)]$ whp.
    \end{enumerate}
\end{lemma}

\begin{proof}
  Note that \itemRef{it:concen-S}, \itemRef{it:concen-B} and \itemRef{it:concen-H} each claim that certain bounds hold whp \emph{simultaneously} for all $t$ in an interval. However, it suffices to show that they hold whp simultaneously for uniformly spaced \emph{anchor points} $t₀ = 0$, $t₁ = Δ$, $t₂ = 2Δ$, … where $Δ > 0$ is a constant. This is because the stated bounds change only by a constant factor in between two consecutive $t_i$. Due to monotonicity of $S(t), B(t)$ and $H(t)$, all bounds that hold for all anchor points also hold up to that constant factor in between the anchor points.

  Since the interval of admissible $t$ has length $𝒪(\log n)$ in all cases it even suffices to focus on any single $t$ as the intersection of $𝒪(\log n)$ events that hold whp still holds whp.
    \begin{enumerate}[(i)]
        • We have $𝔼[S(t)] = km·\Pr[\Exp(1) > t] = km\e^{-t}$ by definition of $S(t)$.
        • Let $i ∈ ℕ₀$. We have $S(i+1) \sim \Bin(S(i),\frac{1}{\e})$ and Hoeffding's inequality (\cref{lem:std-conc} \itemRef{it:hoeffding}) gives ${S(i+1)}/{S(i)} ∈ [\frac 1\e - δ, \frac 1\e + δ]$ with probability $1-\exp(-Ω(S(i)))$ for any constant $δ > 0$. As long as $S(i) = m^{Ω(1)}$ such an event holds whp.
        Note that the largest permitted value $t_{\max} = \frac{9}{10}\ln(m)$ is such that $𝔼[S(t_{\max})] = Ω(m^{\frac{1}{10}})$ so we obtain $S(i) ∈ [km(\frac 1\e-δ)^i,km(\frac 1\e+δ)^i]$ for all $i ∈ [⌈t_{\max}⌉]$ and any constant $δ > 0$ whp by induction. 
        Combined with the preliminary remark this implies the claim where $\pm γ$ in the exponent compensates for $δ$.
        • The claim $B(t) = Ω(P(t))$ follows from the tight link between $B(t)$ and $P(t)$ proved in \cref{lem:coupling-of-Bt}. The proof of $P(t) = Ω(m\e^{-(\frac{k}{k-1}+γ)t})$ is analogous to \itemRef{it:concen-S} since $P(t)$ is a pure death process like $S(t)$ (just with parameter $\frac{k}{k-1}$ instead of $1$). 

        Note that $t_{\max} = \frac{3}{5}\ln(m)$ is again such that $𝔼[P(t_{\max})] = \smash{(k-1)m\e^{-\frac{k}{k-1}\frac{3}{5}\ln(m)}}  = Ω(m^{1-\frac{k}{k-1}·\frac{3}{5}}) = Ω(m^{1-\frac{3}{2}·\frac{3}{5}}) = Ω(m^{\frac{1}{10}})$ is polynomially large in $n$.

        • Consider the $km$ balls before they are assigned to vertices. The probability that a fixed pair of balls will share a vertex at time $t$ is $\frac{1}{n}\e^{-2t}$ because they have to be assigned to the same vertex and both must not have died at time $t$. Defining $H₂(t)$ to be the number of such pairs we have
        \[𝔼[H(t)] ≤ 𝔼[2H₂(t)] = 2\tbinom{km}{2}\tfrac{1}{n}\e^{-2t} = 𝒪(m\e^{-2t}).\]
        • For $t = \frac{3}{5}\ln(m)$ we have $𝔼[H(t)] = 𝒪(m^{1-\frac{6}{5}}) = m^{-1/5}$ by \itemRef{it:exp-H}. We conclude $\Pr[H(t) ≥ 1] ≤ m^{-1/5}$ by Markov's inequality.
        • By the preliminary remark, we can focus on a single $t$. For $t ∈ [\frac{1}{3} \ln(m),\frac{3}{5}\ln(m)]$ we can simply use Markov's inequality:
        \begin{align*}
            \Pr[H(t) > m\e^{-(2-2γ)t}] ≤ \frac{𝔼[H(t)]}{m\e^{-(2-2γ)t}}  = \frac{𝒪(m\e^{-2t})}{m\e^{-(2-2γ)t}}
            = 𝒪(\e^{-2γt}) = 𝒪(m^{-2γ/3}).
        \end{align*}
        Now assume $t ∈ [0,\frac{1}{3}\ln(m)]$. Let $s = 𝒪(m\e^{-(1-γ)t})$ be the upper bound on $S(t)$ that holds whp by \itemRef{it:concen-S}. We have for any $Z = ω(m\e^{-(2-2γ)t})$:
        \begin{align}
            \Pr[H(t) > 2Z] &≤ \Pr[H(t) > 2Z \mid S(t) ≤ s] + \Pr[S(t) > s]\label{eq:boundingHt}\\
            &≤ \Pr[H(t) > 2Z \mid S(t) = s] + m^{-Ω(1)}.\notag
        \end{align}
        In the last step we used that a larger number of balls with assigned lifetimes exceeding $t$ favours a larger number of balls that are not alone at their vertex at time $t$, i.e.\ $\Pr[H(t) > 2Z \mid S(t) = s]$ is monotonically increasing in $s$.
        Let us define $μ := 𝔼[H(t) \mid S(t) = s]$. For the same reasons as in \itemRef{it:exp-H}  we have
        \[μ ≤ 2\tbinom{s}{2}\tfrac{1}{n} = 𝒪(m\e^{-2(1-γ)t}) = 𝒪(m\e^{-(2-2γ)t}) = o(Z).\]
        
        Assume $S(t) = s$ and let $b₁,…,b_s$ be the $s$ balls with assigned lifetimes exceeding $t$. The vertices $v₁,…,v_s ∈ [n]$ that these balls are assigned to are still independent and uniformly random variables.
        Like in \cref{lem:concentration-early} \itemRef{it:CH} we aim to use McDiarmid's inequality (\cref{lem:std-conc} \itemRef{it:mcdiarmid}). Again $H(t)$ can be computed from $v₁,…,v_s$ and changing a single $v_i$ can change $H(t)$ by at most $\pm 2$. We conclude:
        \begin{align*}
            \Pr[H(t) > 2Z &\mid S(t) = s]
            ≤ \Pr\big[|H(t) - μ| > Z \mid S(t) = s\big]
            \stackrel{\textrm{Mc.D.}}{=} 2\exp(-\tfrac{2Z²}{4s})\\
            &= \exp(-Ω(\tfrac{m²\e^{-2(2-2γ)t}}{m\e^{-(1-γ)t}}))
            = \exp(-Ω(m\e^{-3(1-γ)t}))\\
            &= \exp(-Ω(m^{1-(1-γ)}))
            = \exp(-Ω(m^{γ}))
        \end{align*}
        where the step between the second and third line used $t ≤ \frac{1}{3}\ln(m)$. Combining what we have just obtained with \cref{eq:boundingHt} implies the claim. \qedhere
        %
    \end{enumerate}
\end{proof}

We are now in a position to prove \cref{prop:ball-counts} \itemRef{it:few-heavy-balls}.
\begin{proof}[Proof of \cref{prop:ball-counts} \itemRef{it:few-heavy-balls}]
    Conditioned on the high probability events of \cref{lem:concentration-bounds} \itemRef{it:concen-B}, \itemRef{it:concen-H0} and \itemRef{it:concen-H} we have, for some constants $C, C' > 0$ (previously implicit in $𝒪$-notation)
    \begin{align*}
        &&B(t) &≥ Cm\e^{-(k/(k-1)+γ)t} && \text{ for $t ∈ [0,\tfrac{3}{5}\ln(m)]$},&&\\
        &&H(t) &= 0 && \text{ for $t ≥ \tfrac{3}{5}\ln(m)$},&&\\
        &&H(t) &≤ C'm\e^{-(2-2γ)t} && \text{ for $t ∈ [0,\tfrac{3}{5}\ln(m)]$}.&&
    \end{align*}
    Our task is to show $H(t) ≤ \frac{1}{2k}B(t)$ for all $t ≥ t₀$ where $t₀$ is some constant. When $H(t) = 0$ there is nothing to show, so we may assume $t ∈ [0,\frac{3}{5}\ln(m)]$ and $B(t) > 0$. By choosing $γ = 1/12$ we can bound the quotient
    \[ \tfrac{H(t)}{B(t)} ≤ \tfrac{C'}{C} \e^{-(2-2γ-\frac{k}{k-1}-γ)t} = \tfrac{C'}{C} \e^{-(\frac{k-2}{k-1}-3γ)t} ≤ \tfrac{C'}{C} \e^{-(\frac{1}{2}-\frac{3}{12})t} = \tfrac{C'}{C} \e^{-t/4}.\]
    By choosing $t₀ = t₀(k,C,C')$ large enough, we obtain $\frac{H(t)}{B(t)} ≤ \frac{1}{2k}$ for any $t ≥ t₀$ as desired.
\end{proof}




\section{Bounding Peeling Numbers}
\label{sec:bounding-pn}

Now all our previous arguments come together. Recall from \cref{sec:peeling-numbers} that the number of times a given key $x$ moves throughout all insertions is related to the peeling number of the corresponding hyperedge $e$, defined (roughly) as the number of hyperedges that $e$ depends on (directly or indirectly). When aggregating over all $e$, we may conversely count the number of hyperedges that \emph{depend on $e$}. We will show that in our model, where the hypergraph and its random peeling are generated in continuous time, the expected number of such hyperedges increases at most exponentially for a constant amount of time $t₀$. Due to the rareness of heavy balls after time $t₀$, the hyperedges generated after $t₀$ tend to be isolated from each other with only a small number of further dependencies between them. For the precise argument we introduce some new concepts.


\def\Gdep{G_{\mathrm{dep}}}
\def\GdepV{G_{\mathrm{dep}}^V}
\paragraph{The Dependence Graph.}
Recall the definition of direct dependence and peeling numbers from \cref{sec:peeling-numbers}, defined in the context of a peelable hypergraph $H$ and a peeling $F$. Let $\Gdep = (E, \{(e,e') \mid e ∈ E, e' ∈ D_F(e)\})$ be the \emph{dependence graph}. More precisely, each $(e,e')$ with $e ≠ e'$ is contained in the directed \emph{multi}graph $\Gdep$ with the same multiplicity with which $F(e)$ is contained in the \emph{multi}set $e'$. Then $\peel(e)$ is simply the number of non-trivial paths in $\Gdep$ that start at $e$, and $\sum_{e ∈ E}\peel(e)$ is the total number of non-trivial paths in $\Gdep$. We make two modifications before finally bounding this number.

\paragraph{A Vertex-Based Version of the Dependence Graph.} For technical reasons\footnote{We later wish to interact with a vertex $v ∈ V$ without revealing the hyperedge $e$ with $F(e) = v$ (if such a hyperedge exists at all).} we will work with the \emph{vertex-based dependence graph} $\GdepV$ with vertex set $V$. For each $e' ∈ E$ with $F(e') = v'$ each edge $(v,v')$ for $v ≠ v'$ exists with the multiplicity with which $e'$ contains $v$.
While an edge $(e,e')$ in $\Gdep$ indicated “$e'$ must be peeled before $e$” the interpretation of an edge $(v,v')$ in $\GdepV$ is the slightly more subtle “the hyperedge $e'$ with $F(e') = v'$ must be peeled before the hyperedge $e$ with $F(e) = v$ \emph{if such a hyperedge $e$ exists}”.
Any $v ∉ F(E)$ has no incoming edges in $\GdepV$, but may have outgoing edges. It is easy to see that $F$ is a graph homomorphism (though not necessarily a graph isomorphism) from $\Gdep$ to $\GdepV$, i.e.\ the multiplicity of the edge $(e,e')$ in $\Gdep$ equals the multiplicity of the edge $(F(e),F(e'))$ in $\Gdep$. We can define $\peel(v)$ for $v ∈ V$ as the number of non-trivial paths in $\GdepV$ that start at $v$. Then clearly $\peel(e) = \peel(F(e))$ for all $e ∈ E$.

\paragraph{Reversing Dependence.} When counting all paths it does not matter if we group them by start point or end point. Let $p(v)$ be the number of paths in $\GdepV$ that \emph{end} in $v$. Then
\begin{equation}
  \sum_{e∈E}\peel(e) = \sum_{e∈E}\peel(F(e)) ≤ \sum_{v ∈ V} \peel(v) = \sum_{v ∈ V} p(v) - n.\label{eq:pn-to-paths}
\end{equation}
The “$-n$” is due to the inclusion of the trivial path $(v)$ in $p(v)$.

\def\vend{v_{\mathrm{end}}}

\def\cnt{{\scalebox{0.8}{COUNT}}\xspace} 
\def\cntp{{\scalebox{0.8}{COUNT}}${}'$\xspace}

\paragraph{Counting Paths.} The remaining challenge is to bound $𝔼[p(v)]$ for $v ∈ V$. We do so by analysing an algorithm that computes $p(v)$. \Cref{algo:count-paths} (left) shows a generic algorithm \cnt for counting all paths ending in a vertex $\vend ∈ V$ for a given directed acyclic graph $G = (V,E)$ (it should be obviously correct). On the right we give an adaptation \cntp to the special case of $\GdepV$, where $\GdepV$ is implicitly given by a peelable hypergraph $H$ and a peeling $F$ of $H$.
Two remarks are in order.
\begin{figure}[tbp]
  \begin{minipage}{0.57\textwidth}
    \begin{algorithm}[H]
      \algo{\normalfont \cnt{}($G, \vend$)}{
        \For{$v ∈ V$}{
          $p[v] ← 0$
        }
        $p[\vend] ← 1$\;
        \For{$v' ∈ V$ in reverse topological order}{
          \For{$(v,v') ∈ E$}{
            $p[v] ← p[v] + p[v']$
          }
        }
        \Return $\sum_{v ∈ V}p[v]$
      }
    \end{algorithm}
  \end{minipage}
  \begin{minipage}{0.50\textwidth}
    \begin{algorithm}[H]
      \algo{\normalfont \cntp{}($H, F, \vend$)}{
        \For{$v ∈ V$}{
          $p[v] ← 0$
        }
        $p[\vend] ← 1$\;
        \For{$e' ∈ E$ in peeling order}{
          $v' ← F(e')$\;
          \For{$v ∈ e' \setminus \{v'\}$}{
            $p[v] ← p[v] + p[v']$
          }
        }
        \Return $\sum_{v ∈ V}p[v]$
      }
    \end{algorithm}
  \end{minipage}
  \caption[fragile]{On the left: A simple algorithm to compute the number of paths that terminate in a vertex $\vend$ of a directed acyclic graph $G = (V,E)$ using dynamic programming.\\
  On the right: The same algorithm, adapted for $\GdepV$, which is implicitly given by $H = (V,E)$ and $F$.}
  \label{algo:count-paths}
\end{figure}
\begin{description}
  •[\cntp is correct.] The order in which the hyperedges are peeled is a reverse topological ordering of $\Gdep$. We apply the graph homomorphism $F$ to obtain a topological ordering of $\GdepV$, \emph{except} for $V \setminus F(E)$. However, these vertices have indegree $0$ which means two things: They may come last in the reverse topological ordering and the missing iterations of the loop that would be performed by \cnt actually do nothing and may be omitted.
  •[\cntp works in the configuration model.] Executing the $i$-th loop iteration of \cntp need not access any information about $H$ and $F$ other than what the $i$-th round of peeling reveals about $H$ and $F$ anyway.
\end{description}

\noindent
We let \cntp run on $H = H_{n,cn,k}$ for $c = \cpk - ε$, a random peeling $F$ of $H$ and an arbitrary $\vend ∈ V$.
Let $t₀$ be the constant from \cref{prop:ball-counts} \itemRef{it:few-heavy-balls} and $δ = δ(ε,t₀)$ the constant from \cref{prop:ball-counts} \itemRef{it:many-light-balls}. Moreover, let $\mathcal{E} = \mathcal{E}(ε,t₀,δ)$ be the intersection of the corresponding high probability events. We define
\[ T_i := 𝟙[\mathcal{E}]·\sum_{v ∈ V} p_i[v] \]
where $p_i[v]$ is the value of $p[v]$ after the $i$-th iteration of the loop of \cntp for $0 ≤ i ≤ m$, if it exists. The term “$𝟙[\mathcal{E}]$” makes sure that $T_i$ has a well-defined value of $0$ in the unlikely event that $H$ turns out not to be peelable, but without distorting the probability space that conditioning on $\mathcal{E}$ would bring.
We will now apply our results on peeling in the configuration model in continuous time from \cref{sec:peeling-in-config-model}.
\begin{lemma}
  \label{lem:bounding-paths}
    Let $I ∈ [m]$ be the index of the last round of peeling beginning before time $t₀$. 
    \begin{enumerate}[(i)]
      • $𝔼[T_I] ≤ \e^{(k-1)/δ}$.
      • $𝔼[T_m \mid I, (T_j)_{j ≤ I}] ≤ 2kT_I$.
      • $𝔼[p(\vend) \mid \mathcal{E}] ≤ 2k\e^{(k-1)/δ}·(1+o(1))$.
    \end{enumerate}
\end{lemma}

\begin{proof}
  \begin{enumerate}[(i)]
    • Let $i ≤ I$ be the index of a round of peeling beginning at an “early” time $t < t₀$. If $L(t) < δn$ then the high probability event in \cref{prop:ball-counts} \itemRef{it:many-light-balls} fails, which gives $T_I = 0$. Hence, assume $L(t) ≥ δn$.

    At the start of round $i$, a random vertex $v'$ of degree $1$ is picked from the set $V₁^{(i)}$ of $|V₁^{(i)}| = L(t)+1 ≥ δn$ vertices that have degree $1$. In \cntp its current value $p_{i-1}[v']$ is then added to $k-1$ other vertices that form a hyperedge with $v'$. We compute, using that the expectation of $p_{i-1}[v']$ is the mean of $p_{i-1}[v]$ for $v ∈ V₁^{(i)}$:
    \begin{align}
      𝔼[T_i \mid T_{i-1}, i ≤ I] &= T_{i-1} + 𝔼[(k-1)·p_{i-1}[v'] \mid T_{i-1}, i ≤ I]\notag\\
      & = T_{i-1} + (k-1)𝔼\Big[\tfrac{1}{|V₁^{(i)}|}\sum{}_{v ∈ V₁^{(i)}} p_{i-1}[v] \mid T_{i-1}, i ≤ I\Big]\notag\\
      & ≤ T_{i-1} + (k-1)𝔼[\tfrac{1}{δn} T_{i-1} \mid T_{i-1}, i ≤ I] = T_{i-1}·(1+\tfrac{k-1}{δn}).\label{eq:T_i-from-T_i-1-early}
    \end{align}
    Let now $T_i' := T_{\min(i,I)}$ for $i ∈ [m]$.
    Conditioned on $i ≤ I$ we have $T_i' = T_i$ and $T_{i-1}' = T_{i-1}$ so $𝔼[T_i' \mid T_{i-1}', i ≤ I] ≤ T_{i-1}'·(1+\tfrac{k-1}{δn})$ by \cref{eq:T_i-from-T_i-1-early}. Moreover, conditioned on $i > I$ we have $T_i' = T_{i-1}'$ and hence trivially $𝔼[T_i' \mid T_{i-1}', i > I] = T_{i-1}' ≤ T_{i-1}'·(1+\tfrac{k-1}{δn})$. The combination gives $𝔼[T_i' \mid T_{i-1}'] ≤ T_{i-1}'·(1+\tfrac{k-1}{δn})$.
    By induction we obtain
    \begin{align*}
      𝔼[T_i'] &= 𝔼\big[𝔼[T_i' \mid T_{i-1}']\big] ≤ 𝔼[T_{i-1}'·(1+\tfrac{k-1}{δn})]
      = (1+\tfrac{k-1}{δn})·𝔼[T_{i-1}'] \reasonrel{Ind.}{≤} (1+\tfrac{k-1}{δn})^{i},
    \end{align*}
     for all $i ∈ [m]$ where the base case $T₀' = T₀ = 1$ is ensured by initialisation. Using $1+x ≤ \e^x$ we get the desired bound
    \[𝔼[T_I] = 𝔼[T_m'] ≤ (1+\tfrac{k-1}{δn})^{m} ≤ \e^{\frac{k-1}{δ}c} ≤ \e^{(k-1)/δ}.\]
    • Let us call a vertex $v$ \emph{alive} as long as at least one ball remains at $v$ (i.e.\ $v$ is incident to a hyperedge that is yet to be revealed). Let $V_A^{(i)}$ be the set of living vertices after round $i$ and $A_i = 𝟙[\mathcal{E}]·\sum_{v ∈ V_A^{(i)}} p_i[v] ≤ T_i$. We will prove that for all $I ≤ i ≤ m$
    \begin{equation}
      𝔼[T_m \mid I, (T_j)_{j ≤ i}, (A_j)_{j ≤ i}] ≤ T_i + (2k-1)A_i \label{eq:Tm-From-Ti}      
    \end{equation}
    which implies the claim when plugging in $i = I$ and the worst case $A_I = T_I$.

    Let $i > I$ be the index of a round beginning at time $t ≥ t₀$. The peeling process first selects a random light ball $b₀$ at some vertex $v₀$ for removal, then selects $k-1$ random balls $b₁,…,b_{k-1}$ at vertices $v₁,…,v_{k-1}$ and emits the hyperedge $\{v₀,v₁,…,v_{k-1}\}$. Let $ξ₁,…,ξ_{k-1} ∈ \{0,1\}$ indicate whether $b₁,…,b_{k-1}$ are heavy at time of selection. For $w := p_{i-1}[v₀]$ the effects of round $i$ can be described as
    \begin{align}
      T_i = T_{i-1} + (k-1)w && \text{and} &&
      A_i ≤ A_{i-1} - w + \sum_{i = 1}^{k-1}ξ_iw \label{eq:next-Ti-Ai}
    \end{align}
    where in the formula for $A_i$ the term $-w$ accounts for $v₀$ itself dying and $ξ_i·w$ accounts for the added weight at vertex $v_i$, which only remains alive if the ball removed from it was not the last one, i.e.\ if $b_i$ was heavy. Due to the inclusion of the high probability events from \cref{prop:ball-counts} in $\mathcal{E}$, we may assume $H(t) ≤ \frac{1}{2k}B(t)$ for all $t ≥ t₀$. This ensures that $\Pr[ξ_i = 1] ≤ \frac{1}{2k}$ for all $i ∈ [k-1]$ and regardless of any events prior to round $i$. Therefore, for any $i > I$
    \begin{align*}
      𝔼&[T_i+(2k-1)A_i \mid I, (T_j)_{j<i}, (A_j)_{j<i},w]\\
      &\eqRefRel{eq:next-Ti-Ai}{}{≤} 𝔼\big[T_{i-1}+(k-1)w+(2k-1)(A_{i-1}-w + \sum_{i = 1}^{k-1}ξ_iw)  \mid I, (T_j)_{j<i}, (A_j)_{j<i}, w \big]\\
      &≤ T_{i-1}+(2k-1)A_{i-1}+w·\Big((2k-1)·𝔼\Big[\sum_{i = 1}^{k-1}ξ_i  \mid I, (T_j)_{j<i}, (A_j)_{j<i}, w \Big] -k \Big)\\
      &≤ T_{i-1}+(2k-1)A_{i-1}+w·\Big((2k-1)·\tfrac{k-1}{2k} -k \Big) ≤ T_{i-1}+(2k-1)A_{i-1}.
    \end{align*}
    Since the upper bound does not depend on $w$ we conclude
    \begin{equation}
      𝔼[T_i+(2k-1)A_i \mid I, (T_j)_{j<i}, (A_j)_{j<i}] ≤ T_{i-1}+(2k-1)A_{i-1}.\label{eq:T_i-from-T_i-1}
    \end{equation}
    We finally obtain \cref{eq:Tm-From-Ti} via induction working backwards from the trivial base case $i = m$ (where $T_i = T_m$ and $A_i = 0$).
    \begin{align*}
      𝔼[T_m \mid & I, (T_j)_{j<i}, (A_j)_{j<i}] ≤ 𝔼\big[𝔼[T_m \mid I, (T_j)_{j≤i}, (A_j)_{j≤i}] \mid I, (T_j)_{j<i}, (A_j)_{j<i}\big]\\
      & \reasonrel{Ind.}{≤}
      𝔼[T_i+(2k-1)A_i \mid I, (T_j)_{j<i}, (A_j)_{j<i}]
      \eqRefRel{eq:T_i-from-T_i-1}{}{≤} T_{i-1}+(2k-1)A_{i-1}.
    \end{align*}
    • We simply put together \lipicsLabel{(i)} and \lipicsLabel{(ii)}:
    \begin{align*}
      𝔼[p(\vend) \mid \mathcal{E}] & = \frac{𝔼[𝟙[\mathcal{E}]·p(\vend)]}{\Pr[\mathcal{E}]} = (1+o(1))𝔼[T_m]
      = (1+o(1))𝔼\big[𝔼[T_m \mid I, (T_j)_{j ≤ I}]\big]\\
      &\reasonrel{\lipicsLabel{(ii)}}{≤}
       (1+o(1))𝔼[2kT_I]
      \reasonrel{\lipicsLabel{(i)}}{≤}
      (1+o(1))2k\e^{(k-1)/δ}.\qedhere
    \end{align*}
  \end{enumerate}
\end{proof}

\begin{proof}[Proof of \cref{prop:peeling-numbers-small}]
  The fact that $H$ is peelable follows from our choice for $\mathcal{E}$. The argument is completed by letting each vertex take the role of $\vend$ in \cref{lem:bounding-paths}:
  \begin{align*}
    𝔼\Big[\sum_{e ∈ E} \peel_F(e) \mid \mathcal{E}\Big] &\eqRefRel{eq:pn-to-paths}{Eq. }{≤} 𝔼\Big[\sum_{v ∈ V} p(v) \mid \mathcal{E}\Big]
    ≤ \sum_{v ∈ V} 𝔼[p(v) \mid \mathcal{E}]\\
    & \refrel{lem:bounding-paths}{Lem.}{≤} |V|·2k\e^{(k-1)/δ}·(1+o(1)) = 𝒪(n).\qedhere
  \end{align*}  
\end{proof}

\fi 

\section{Conclusion and Future Work}

This paper proves $𝒪(1)$ expected amortised running times for random walk insertions into cuckoo hash tables and is the first to yield meaningful results for small values of $k$ such as $k = 3$.
Our proof strategy is to link the number of times that a key $x$ moves to the number of times that certain other keys move, where these other keys all preceed $x$ in the peeling process. The main technical challenge
\iffullversion\else
 (addressed in the full version of this paper)
\fi
 was to extend an existing analysis of this peeling process in order to obtain stronger guarantees on its late stages when a sublinear number of keys remain.

To make the result more relevant to practitioners, it is natural to pursue a generalisation to sequences of insertions \emph{and deletions} and to buckets of size $ℓ ≥ 2$. The author suspects that the given argument can be correspondingly extended with moderate technical complications.


The more important goal, however, is to obtain a result that works up to the load threshold (for all $c < \ck-ε$), not just up to the peeling threshold (for $c < \cpk-ε$). There is at least one reason for optimism, namely the recent discovery of a variant of cuckoo hashing that raises the peeling threshold to the load threshold \cite{W:SpatialCoupling:2021}. Briefly, a key's $k$ hashes are randomly distributed in a random window of $γn$ consecutive buckets. The peeling threshold of this variant is equal to $\ckl - ε$ where $ε(γ)$ can be made arbitrarily small. However, when using this variant, an analysis can no longer rely on the configuration model due to a lack of symmetry between the vertices, meaning that even if the general idea is still sound, the proof would have to use different methods.

Regardless of whether such an improvement is achievable, we believe this paper to be a promising step forward in the ongoing project of retrofitting the widespread use of cuckoo hash tables and cuckoo filters with strong theoretical guarantees.

\ifanonymous\else
  \paragraph{Acknowledgements.} I would like to thank Martin Dietzfelbinger for providing several useful comments that helped with improving the presentation of this paper as well as an anonymous reviewer who gave useful feedback regarding the technical argument.
\fi

\bibliographystyle{plainurl}
\bibliography{bibliographie}

\iffullversion\else
  \nolinenumbers
  \pagestyle{empty}
  \newpage
  \ 
  
  \vfill
    {\centering\sffamily
      {\Huge\bfseries – Full Version –\\}
      {\Huge (includes Sections 4 + 5)\\}
    }
  \vfill
  \foreach\i in {1,...,21}{
    \newpage
    \begin{tikzpicture}[overlay]
      \pgfmathsetmacro{\odd}{int(Mod(\i,2))}
      \ifnum\odd=1
        \node[shift={(19.5em,-31.3em)}] {\includegraphics[page=\i]{fullversion.pdf}};    
      \else
        \node[shift={(17.5em,-31.3em)}] {\includegraphics[page=\i]{fullversion.pdf}};    
      \fi
    \end{tikzpicture}
  }
  
\fi



\end{document}